\documentclass[]{elsarticle}
\usepackage{amsmath,amssymb,amsthm}
\usepackage{graphicx}
\usepackage{a4wide}
\usepackage[ruled, vlined, linesnumbered]{algorithm2e}

\usepackage{microtype}

\newcommand{\NP}{\textsf{NP}}

\newtheorem{theorem}{Theorem}[]
\newtheorem{lemma}[theorem]{Lemma}

\newtheorem{proposition}[theorem]{Proposition}
\newtheorem{observation}[theorem]{Observation}

\def\mwis{{\sc{MWIS }}}
\def\minwed{{\sc{Min-WED }}}
\def\maxwed{{\sc{Max-WED }}}

\begin{document}
\begin{frontmatter}
\title{Polynomial-time Algorithms for Weighted Efficient Domination Problems in AT-free Graphs and Dually Chordal Graphs}

\author[ROSTOCK]{Andreas Brandst\"adt}
\ead{ab@informatik.uni-rostock.de}

\author[IAM,FAMNIT]{Pavel Fi\v cur}
\ead{pavel.ficur@upr.si}

\author[KENT]{Arne Leitert}
\ead{aleitert@cs.kent.edu}

\author[IAM,FAMNIT]{Martin Milani\v c\corref{cor}}
\ead{martin.milanic@upr.si}

\address[ROSTOCK]{Institut f\"ur Informatik, Universit\"at Rostock, Germany}
\address[IAM]{University of Primorska, UP IAM, Muzejski trg 2, SI6000 Koper, Slovenia}
\address[FAMNIT]{University of Primorska, UP FAMNIT, Glagolja\v ska 8, SI6000 Koper, Slovenia}
\address[KENT]{Kent State University, Department of Computer Science, Kent, Ohio 44242, USA}

\cortext[cor]{Corresponding author}

\begin{abstract}
An efficient dominating set (or perfect code) in a graph is a set of vertices the closed neighborhoods of which partition the vertex set of the graph.
The minimum weight efficient domination problem is the problem of finding an efficient dominating set of minimum weight in a given vertex-weighted graph;
the maximum weight efficient domination problem is defined similarly.
We develop a framework for solving the weighted efficient domination problems based on a reduction to the maximum weight independent set problem in the square of the input graph. 
Using this approach, we improve on several previous results from the literature by deriving polynomial-time algorithms for the weighted efficient domination problems in the classes of dually chordal and AT-free graphs. 
In particular, this answers a question by Lu and Tang regarding the complexity of the minimum weight efficient domination problem in strongly chordal graphs.
\end{abstract}

\end{frontmatter}

\noindent
\textbf{Notice:} This is the author's version of a work that was accepted for publication in Information Processing Letters. Changes resulting from the publishing process, such as peer review, editing, corrections, structural formatting, and other quality control mechanisms may not be reflected in this document. Changes may have been made to this work since it was submitted for publication. A definitive version was subsequently published in Information Processing Letters 115 (2015) 256-262, DOI: 10.1016/j.ipl.2014.09.024.

\section{Introduction}

The concept of an efficient dominating set in a graph was introduced by Biggs~\cite{Biggs} as a generalization of the notion of a perfect error-correcting code in coding theory. Given a (simple, finite, undirected) graph $G=(V,E)$, we say that a vertex \emph{dominates} itself and each of its neighbors. 
An \emph{efficient dominating set} in $G$ is a subset of vertices $D\subseteq V$ such that every vertex $v \in V$ is dominated by precisely one vertex from $D$. 
Efficient domination has several interesting applications in coding theory and resource allocation of parallel processing systems~\cite{Biggs, LLT97, LS90}. The notion of an efficient dominating set appeared in the literature under various other names such as: \emph{perfect code}, \emph{$1$-perfect code}, \emph{independent perfect dominating set}, and \emph{perfect dominating set}. 
Note, however, that the name \emph{perfect dominating set} has also been used in the literature to denote a subset of vertices $D\subseteq V$ such that every vertex $v \in V\setminus D$ is dominated by precisely one vertex from $D$. 
See~\cite{LT02} for a nice historical overview of the notion of efficient dominating set.

A graph is \emph{efficiently dominatable} if it contains an efficient dominating set.
All paths are efficiently dominatable, and a cycle $C_k$ on $k$ vertices is efficiently dominatable if and only if $k$ is a multiple of $3$.
Bange \emph{et al.}~\cite{BBS88} showed that if a graph $G$ has an efficient dominating set, then all efficient dominating sets of $G$ have the same cardinality, which equals the minimum cardinality of a dominating set of $G$. 
The efficient domination (ED) problem consists in determining whether the input graph is efficiently dominatable. 
The ED problem is \NP-complete even for restricted graph classes such as planar cubic graphs~\cite{Kratochvil91}, bipartite graphs~\cite{YL96}, planar bipartite graphs~\cite{LT02}, chordal bipartite graphs~\cite{LT02}, chordal graphs~\cite{YL96}, and line graphs of planar bipartite graphs of maximum degree three~\cite{BHN}.
On the other hand, the ED problem is polynomial for several graph classes, including trees~\cite{BBS88,FH91}, block graphs~\cite{YL96}, interval graphs~\cite{CL94, CRC95, KE00, KMM95}, circular-arc graphs~\cite{CL94,KE00}, cocomparability graphs~\cite{C97,CRC95}, bipartite permutation graphs~\cite{LT02}, permutation graphs~\cite{LLT97}, distance-hereditary graphs~\cite{LT02}, trapezoid graphs~\cite{LLT97,L98}, split graphs~\cite{CL93}, dually chordal graphs~\cite{BLR}, AT-free graphs~\cite{BLR,BKM99}, and hereditary efficiently dominatable graphs~\cite{Milanic}.
The efficient domination problem has also been studied from a parameterized point of view, see, e.g.,~\cite{Cesati,Guo-Niedermeier}.

In this paper, we consider two weighted versions of the ED problem.
In its decision form, the minimization version of problems can be stated as follows:

\medskip
\begin{center}
\fbox{\parbox{0.85\linewidth}{\noindent
{\sc Minimum Weight Efficient Dominating Set (Min-WED)}\\[.8ex]
\begin{tabular*}{.93\textwidth}{rl}
\emph{Input:} & A graph $G$, vertex weights $w:V\to \mathbb{Z}$, an integer $k$.\\
\emph{Question:} &  Does $G$ contain an efficient dominating set $D$ \\ & of total weight $w(D) := \sum_{x\in D}w(x)\le k$?
\end{tabular*}
}}
\end{center}
\medskip

The maximization version of problem ({\sc Maximum Weight Efficient Dominating Set} \hbox{\sc (Max-WED)}), can be defined analogously, replacing the condition $w(D) \le k$ with $w(D) \ge k$.

Clearly, a graph $G=(V,E)$ contains an efficient dominating set if and only if $(G,w,|V|)$ is a yes instance to the \minwed problem, where $w(x) = 1$ for all $x\in V$.
Consequently, the \minwed problem is \NP-complete in every class of graphs where the ED problem is \NP-complete.
On the other hand, the \minwed problem is solvable in polynomial time for trees~\cite{Y92}, cocomparability graphs~\cite{C97,CRC95}, split graphs~\cite{CL93}, interval graphs~\cite{CL94, CRC95}, circular-arc graphs~\cite{CL94}, permutation graphs~\cite{LLT97}, trapezoid graphs~\cite{LLT97,L98}, bipartite permutation graphs~\cite{LT02}, convex bipartite graphs~\cite{BLR} and their superclass interval bigraphs~\cite{BLR}, distance-hereditary graphs~\cite{LT02}, block graphs~\cite{YL96} and hereditary efficiently dominatable graphs~\cite{CMR00,Milanic}.
Since negative weights are allowed, the \maxwed problem is equivalent to the \minwed problem.

We develop a framework for solving the \minwed and \maxwed problems based on a reduction to the {\sc Maximum Weight Independent Set} problem in the square of the input graph (this is done in Section~\ref{sec:wed-mwis}). 
We then apply this framework, together with some existing results from the literature, to derive new polynomial cases of the \minwed problems, namely the classes of dually chordal graphs and AT-free graphs (in Sections~\ref{sec:dually-chordal} and~\ref{sec:AT-free}, respectively).
The class of dually chordal graphs contains the class of strongly chordal graphs, for which the existence of a polynomial-time algorithm for the \minwed problem was posed as an open problem by Lu and Tang in~\cite{LT02}.
We give a linear-time algorithm for the \minwed and \maxwed problems in the class of dually chordal graphs.
Our algorithm for the \minwed problem in the class of AT-free graphs is of complexity ${\mathcal O}\big(\min\{nm+n^2,n^\omega\}\big)$, where $\omega < 2.3727$ is the matrix multiplication exponent~\cite{Williams}, and $n$ and $m$ denote the number of vertices and edges of the input graph, respectively.
This improves on the existing polynomial-time algorithms for the ED problem in AT-free graphs~\hbox{\cite{BLR,BKM99}}, both of which run in time ${\mathcal O}(n^4)$.

\bigskip
In Fig.~\ref{fig:classes} below, we show the Hasse diagram of the poset of most of the graph classes mentioned above, ordered with respect to inclusion.
For each class, we state the complexity of the \minwed  problem, denoting by \NP-c the fact that the problem is \NP-complete in the corresponding class, while in the case of polynomial-time solvability, we state the running times of the fastest known algorithms.
The inclusion relations in the figure were verified with help of the Information System on Graph Classes and their Inclusions~\cite{graphclasses}.

\begin{figure}[ht]
\begin{center}
\includegraphics[width=\linewidth]{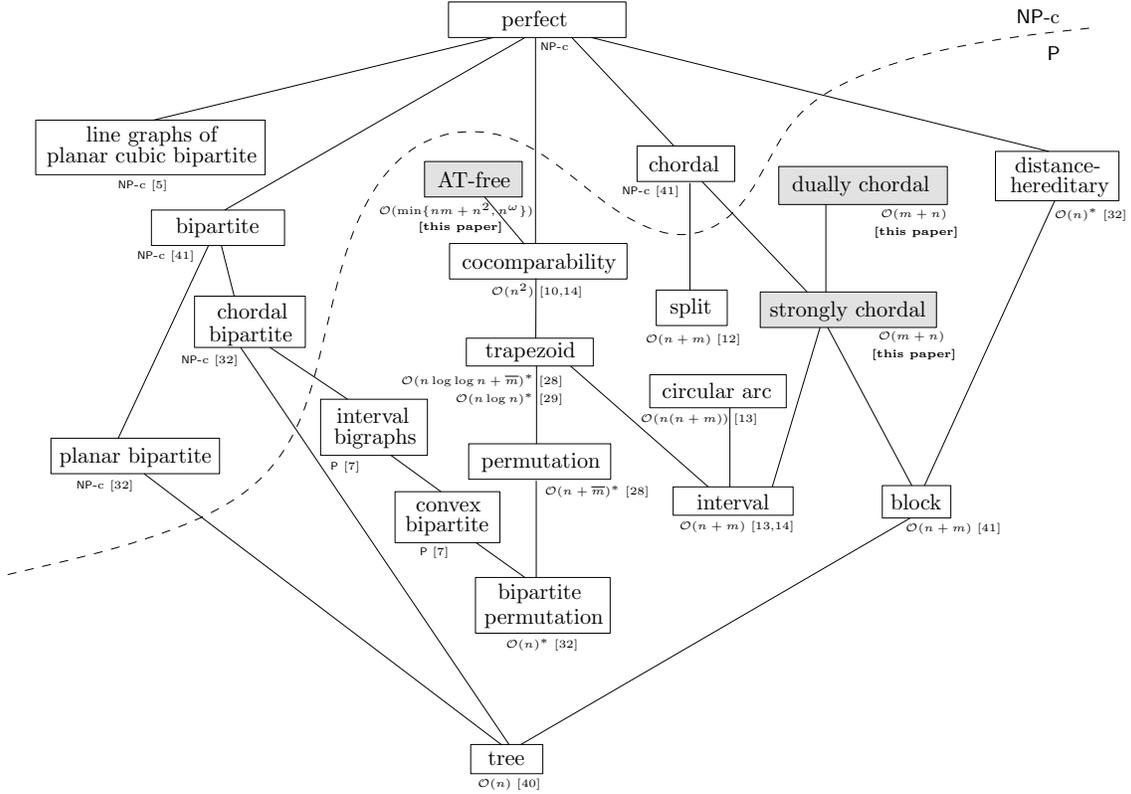}
\caption{
Algorithmic complexity of the weighted efficient domination problem in various graph classes.
In the time complexities marked by $^*$, it is assumed that a graph is given by a trapezoid diagram (in the case of trapezoid graphs), by a permutation (in the case of permutation and bipartite permutation graphs), or by a one-vertex-extension ordering (in the case of distance-hereditary graphs).
}\label{fig:classes}
\end{center}
\end{figure}

\section{Preliminaries}

We only consider finite, simple and undirected graphs.
As usual, the \emph{neighborhood} of a vertex $v$ in a graph $G = (V,E)$ is denoted by $N_G(v) := \{u\in V\mid uv\in E\}$ (or simply $N(v)$, if the graph is clear from the context), and the closed neighborhood of $v$ is $N_G[v]:= N_G(v)\cup \{v\}$ (or simply $N[v]$).
The \emph{degree} of a vertex $v$ in a graph $G$ is $\deg_G(v) := |N_G(v)|$.
The \emph{square} of a graph $G = (V,E)$ is the graph $G^2 = (V,E^2)$ such that $uv\in E^2$ if and only if either $uv\in E$ or $u$ and $v$ are distinct and have a common neighbor in $G$.
The \emph{complement} of a graph $G = (V,E)$ is the graph $\overline{G}= (V,\overline{E})$ such that two distinct vertices $u$ and $v$ of $G$ are adjacent in $\overline{G}$ if and only if they are non-adjacent in $G$.
An \emph{independent set} in a graph is a set of pairwise non-adjacent vertices, and a \emph{clique} is a set of pairwise adjacent vertices. 
Given a graph $G$ and a total ordering $\sigma=(v_1,\ldots, v_n)$ of its vertex set, we will denote by
$G_{\sigma,i}$ the subgraph of $G$ induced by $\{v_i,v_{i+1},\ldots, v_n\}$, and by $N_{\sigma,i}(v)$ (resp.~$N_{\sigma,i}[v]$), the neighborhood (resp., the closed neighborhood) of a vertex $v\in V(G_{\sigma,i})$ in $G_{\sigma,i}$.

\emph{Chordal graphs} are graphs in which every cycle on at least $4$ vertices has a chord (an edge connecting two non-consecutive vertices on the cycle). 
It is well known that chordal graphs are precisely the graphs that admit a perfect elimination ordering. 
A  \emph{perfect elimination ordering} of a graph $G$ is a total ordering $\sigma = (v_1,\ldots, v_n)$ of its vertex set such that for every $i\in \{1,\ldots, n\}$, vertex $v_i$ is simplicial in $G_{\sigma,i}$, that is, the set $N_{\sigma,i}(v_i)$ is a clique in $G_{\sigma,i}$ (equivalently: in $G$).
\emph{Dually chordal graphs} were introduced in~\cite{BDCV} as graphs that admit a certain vertex ordering called a maximum neighborhood ordering.
Given two vertices $u$ and $v$ in a graph $G = (V,E)$, vertex $u\in N[v]$ is a \emph{maximum neighbor} of vertex $v$ if $N[w] \subseteq N[u]$ holds for all $w\in N[v]$. 
A linear ordering $\sigma=(v_1,\ldots, v_n)$ of $V$ is a \emph{maximum neighborhood ordering} of $G$ if for all $i\in \{1,\ldots, n\}$, vertex $v_i$ has a maximum neighbor in the graph $G_{\sigma,i}$.
Dually chordal graphs admit several equivalent characterizations and form a generalization of \emph{strongly chordal graphs}, a well-known subclass of chordal graphs properly containing the classes of trees and interval graphs.
In fact, strongly chordal graphs are exactly the hereditary dually chordal graphs, that is, graphs for which each induced subgraph is a dually chordal graph. Algorithmic aspects of dually chordal graphs were treated systematically in~\cite{BCD98}.

A partial order may be viewed as a transitive directed acyclic graph. 
A \emph{comparability graph} is the graph obtained by ignoring the edge directions of a transitive directed acyclic graph. 
A graph is \emph{cocomparability} if its complement is a comparability graph.

An \emph{asteroidal triple} in a graph $G$ is a subset $I = \{u,v,w\}$ of three pairwise non-adjacent vertices such that for every vertex $x\in I$, the two vertices in $I\setminus \{x\}$ are contained in the same connected component of the graph $G-N[x]$. 
A graph $G$ is said to be \emph{AT-free}~\cite{COS97} if it contains no asteroidal triples.
AT-free graphs form a large class of graphs containing cographs, interval graphs, permutation graphs, trapezoid graphs, and cocomparability graphs. For further details on graph classes, see~\cite{BLV99,Gol04}.

\section{The Reduction}
\label{sec:wed-mwis}

The decision version of the \mwis problem can be stated as follows:

\medskip
\begin{center}
\fbox{\parbox{0.85\linewidth}{\noindent
{\sc Maximum Weight Independent Set (MWIS)}\\[.8ex]
\begin{tabular*}{.93\textwidth}{rl}
\emph{Input:} & A tuple $(G,w,k)$ consisting of a graph $G = (V,E)$, \\ & vertex weights $w:V\to \mathbb{Z}$, and an integer $k\in \mathbb{Z}$.\\
\emph{Question:} & Does $G$ contain an independent set $I\subseteq V$ such that $w(I)\ge k$?
\end{tabular*}
}}
\end{center}
\medskip

This problem, together with its unweighted version, is a classical \NP-complete problem, which remains hard even for several restricted graph classes such as triangle-free graphs~\cite{Pol74}, $K_{1,4}$-free graphs \cite{Min80} and planar graphs of maximum degree at most three~\cite{GJ77}.
On the other hand, the MWIS problem has been shown to admit polynomial solutions in several graph classes, such as claw-free graphs~\cite{Min80,NT01,OPS08}, and their generalization fork-free graphs~\cite{LM08}, perfect graphs~\cite{GLS84}, and AT-free graphs~\cite{BKM99}.

The following simple but useful observation from~\cite{BLR} establishes a connection between efficient dominating sets in a graph $G$ and independent sets in its square.

\begin{observation}[\cite{BLR}]\label{lem:eds-wis}
Let $G=(V,E)$ be a graph, and let $w$ be a vertex weight function for $G$ defined by $w(x) = |N_G[x]|$ for all $x\in V$. 
Then, the following statements are equivalent for every $I\subseteq V$:
\begin{enumerate}
  \item[(i)] $I$ is an efficient dominating set in $G$.
  \item[(ii)] $I$ is a maximum-weight independent set in $G^2$ such that $w(I) = |V|$.
\end{enumerate}
\end{observation}

The above observation (as well as a slightly more general observation from~\cite{Milanic}) has an immediate algorithmic corollary, reducing the ED problem to the MWIS problem in the square of the input graph. 
In~\cite{BLR,Milanic}, the exact time complexity of such a reduction was not specified; we do this in the following proposition.
Given a graph $G$, we denote by $|G|$ its encoding length.

\begin{proposition}[\cite{BLR,Milanic}]\label{prop:1}
Let $\cal C$ be a graph class for which the \mwis problem is solvable in time $T(|G|)$ on squares of graphs from $\cal C$.
Then, the ED problem for ${\cal C}$ is solvable in time \hbox{${\mathcal O}\big(\min\{nm+n,n^\omega\}+T(|G^2|))$}.
\end{proposition}

\begin{proof}
It is easy to see that the square $G^2$ can be computed in time ${\mathcal O}(n^{\omega})$ using matrix multiplication. 
Alternatively, assuming that $G$ is given by adjacency lists, computing $G^2$ can be done in time ${\mathcal O}(nm{+n})$, as follows. 
First, we fix an ordering $\sigma = (v_1,\ldots, v_n)$ of $V$ and order the adjacency lists with respect to $\sigma$.
This can be done in linear time, see, e.g.,~\cite[p.~$36$]{Gol04}.
Then, we create a new copy of these lists. For every vertex $v\in V$, we process its neighbors in order and for each of them, say $w$, we merge the current adjacency list of $v$ with the adjacency list of $w$ (removing duplicates).
Since the lists are ordered, each such merging step can be done in ${\mathcal O}(n)$ time.
Hence, the overall time complexity of computing $G^2$ is
\[{\mathcal O}(m+n)+\sum_{v\in V}\sum_{w\in N_G(v)}{\mathcal O}(n) = {\mathcal O}(m+n)+{\mathcal O}(n)\cdot\sum_{v\in V}\deg_G(v) {= {\mathcal O}(nm+n+m)} = {\mathcal O}(nm+n)\,,\]
where the last equality holds since $n\ge 1$.
This justifies the time complexities given in Proposition~\ref{prop:1}.
\end{proof}

We extend the approach of Brandst\"adt {\it et al.}~\cite{BLR} and Milani\v c~\cite{Milanic} to the weighted versions of the ED problem based on a suggestion made in~\cite{Leitert}.
We have the following

\begin{theorem}\label{thm:wed-mwis}
Let $\cal C$ be a graph class for which the \mwis problem is solvable in time $T(|G|)$ on squares of graphs from $\cal C$.
Then, the \minwed and \maxwed problems are solvable on graphs in ${\cal C}$ in time ${\mathcal O}\big(\min\{nm+n,n^\omega\}+T(|G^2|))$.
\end{theorem}

The proof of this theorem will rely on two auxiliary lemmas.
The first one shows how the problem can be reduced to the case of non-negative weights.

\begin{lemma}\label{lem:wed-mwis}
Let $(G = (V,E), w, k)$ be an instance to the \minwed problem such that $\mu = \min\left\{\min\{w(x)\mid x\in V\},k\right\}<0$.
Suppose that $G$ has an efficient dominating set, and let $\gamma$ be the common cardinality of all efficient dominating sets of $G$.
For all $x\in V$, let $\tilde w(x) = w(x)+|\mu|$, and let $\tilde k = k+|\mu|\cdot\gamma$.
Then, $(G, w, k)$  is a yes instance to the \minwed problem if and only if $(G, \tilde w, \tilde k)$ is a yes instance to the \minwed problem.
In addition, $\tilde \mu = \min\left\{\min\{\tilde w(x)\mid x\in V\},\tilde k\right\}\ge 0$.
\end{lemma}

\begin{proof}
Recall that if a graph $G$ has an efficient dominating set, then all efficient dominating sets of $G$ have the same cardinality, which equals the minimum cardinality of a dominating set of $G$. 
This implies that if we replace each $w(x)$ with $\tilde w(x) = w(x)+|\mu|$, then the weight of each efficient dominating set will increase by exactly $|\mu|\gamma$, where $\gamma$ is the common cardinality of all efficient dominating sets of $G$.
Consequently, $(G,w,k)$ is a yes instance if and only if $(G,\tilde w,k+|\mu|\gamma)$ is a yes instance.
The fact that $\tilde \mu\ge 0$ is clear.
\end{proof}

The second lemma deals with the case of non-negative weights.

\begin{lemma}\label{lem:wed-mwis-2}
Let $(G = (V,E), w, k)$ be an instance to the \minwed problem such that $0 \leq \min\{ \min\{ w(x) \mid x \in V \}, k \}$.
Let $M = \max\{\sum_{x\in V}w(x),k\}+1$, let $w'(x) = M\cdot |N_G[x]| - w(x)$, for all $x\in V$, and $k' = M\cdot |V|-k$.
Then, $(G, w, k)$  is a yes instance to the \minwed problem if and only if $(G^2, w', k')$ is a yes instance to the \mwis problem.
\end{lemma}

\begin{proof}
On the one hand, if $D$ is an efficient dominating set in $G$ such that $w(D)\le k$, then, by Observation~\ref{lem:eds-wis}, $D$ is an independent set in $G^2$ such that $\sum_{x\in D}|N_G[x]| = |V|$.
Therefore, 
\[w'(D) = M\cdot \sum_{x\in D}|N_G[x]| - \sum_{x\in D}w(x)\ge M\cdot |V|-k = k'\,,\]
and $(G^2, w', k')$ is a yes instance to the \mwis problem.

On the other hand, let $I$ be an independent set in $G^2$ with $k'\le w'(I)$.
We claim that $\sum_{x\in I}|N_G[x]| = |V|$ and $w(I)\le k$.
Since $I$ is an independent set in $G^2$, it is an independent set in $G$ such that the closed neighborhoods (in $G$) of its elements are pairwise disjoint. 
Therefore, $\sum_{x\in I}|N_G[x]| \le |V|\,.$
Conversely, $k'\leq w'(I)$ implies $\sum_{x\in I}|N_G[x]|\geq \frac{k'}{M}=|V|-\frac{k}{M}>|V|-1,$ since $k<M$.
We thus have $\sum_{x\in I}|N_G[x]| = |V|$.
By Observation~\ref{lem:eds-wis}, $I$ is an efficient dominating set in $G$.
The inequality $w'(I)\ge k'$ is equivalent to
\[ M\cdot \sum_{x\in I}|N_G[x]| - \sum_{x\in I}w(x) \ge M\cdot |V|-k\,,\] 
which implies $w(I) = \sum_{x\in I}w(x)\le k$.
Thus, $(G, w, k)$ is a yes instance to the \minwed problem, as claimed.
\end{proof}

\begin{proof}[Proof of Theorem~\ref{thm:wed-mwis}]
Consider an instance $(G = (V,E), w, k)$ with $|V| = n$ and $|E| = m$ to the \minwed problem.
Let $\mu = \min\left\{\min\{w(x)\mid x\in V\},k\right\}$.
If $\mu <0$, we transform the instance, using Lemma~\ref{lem:wed-mwis} to an equivalent instance $(G, \tilde w, \tilde k)$ to the \minwed problem with non-negative weights and $\tilde k$.
The time complexity of this step is dominated by solving the ED problem in $G$, which, by Proposition~\ref{prop:1}, is of the order \hbox{${\mathcal O}\big(\min\{nm+n,n^\omega\}+T(|G^2|))$}.

To solve the \minwed problem on $(G, \tilde w, \tilde k)$, we apply Lemma~\ref{lem:wed-mwis-2}. 
The corresponding instance $(G^2, w', k')$ to the \mwis problem can be computed in time $O(n+m)$.

Summarizing, in order to solve the \minwed problem, we only need to compute $G^2$ and solve at most two instances of the \mwis problem on $G^2$.

The \maxwed problem can be reduced in time ${\mathcal O}(n+m)$ to the \minwed problem, since a tuple $(G,w,k)$ is a yes instance to the \maxwed problem if and only if $(G,-w,-k)$ is a yes instance to the \minwed problem.
Thus the result follows.
\end{proof}

\section{New Polynomial Cases of the Weighted Efficient Domination Problems}

In this section, we exploit the approach developed in Section~\ref{sec:wed-mwis} and develop new polynomial results for the \minwed and \maxwed problems.

\subsection{Dually chordal graphs}\label{sec:dually-chordal}

In~\cite{LT02}, Lu and Tang wrote that {\it ``$(\ldots)$ it would be of interest to know whether or not there is a polynomial-time algorithm to solve the weighted efficient domination problem on} ($\ldots$) {\it strongly chordal graphs.''}
Recently, Brandst\"adt \emph{et al.}~\cite{BLR} gave a linear-time algorithm for the efficient domination problem in the class of dually chordal graphs.
The algorithm is a modification of Frank's algorithm (Algorithm~\ref{algo:mwisChordal} below), which solves the \mwis problem for chordal graphs in linear time~\cite{Frank}. 
The modification allows to find a maximum weight independent set of~$G^2$ in linear time if $G$ is dually chordal, without computing its square.

\begin{algorithm}
\caption{\cite{Frank} Algorithm to find a maximum weight independent set in chordal graphs.}\label{algo:mwisChordal}

\KwIn{A chordal graph $G=(V,E)$ with $|V| = n$ and a vertex weight function~$\omega$.}
\KwOut{A maximum weight independent set $I$ of $G$.}

Find a perfect elimination ordering $\sigma = (v_1,\ldots,v_n)$ of $G$ and set $I \leftarrow \emptyset$. \label{line:mwisChorFindPEO}

\For{$i =  1$ \KwTo $n$}
{
    If $\omega(v_i) > 0$, mark $v$ and set $\omega(u) \leftarrow  \max(\omega(u) - \omega(v_i), 0)$ for all vertices $u \in N_{\sigma,i}(v_i)$.
}

\For{$i= n~{\bf downto}~1$}
{
    If $v_i$ is marked, set $I \leftarrow I \cup \{v_i\}$ and unmark all $u \in N_{\{v_1,\ldots, v_{i-1}\}}(v_i)$.
}

\Return{$I$}
\end{algorithm}

We will make a similar modification of Frank's algorithm so that the obtained algorithm (Algorithm~\ref{algo:EDdc} below) solves the \minwed problem for dually chordal graphs. 
The modification contains two major changes. 
First, we add a preprocessing step, transforming the given \minwed instance (with possibly negative weights) to an instance of the \mwis problem. 
Second, we rewrite Algorithm~\ref{algo:mwisChordal} to find a maximum weight independent set of the square of the given graph. 
By Lemma~\ref{lem:wed-mwis-2}, this will be a solution to the \minwed problem.

For the first step, we use Lemma~\ref{lem:wed-mwis}. 
For the second, we need the two following lemmas.

\begin{lemma}[\cite{BCD98}]\label{lem:MNO_linear}
    A maximum neighborhood ordering of $G$ which simultaneously is a perfect elimination ordering of $G^2$ can be found in linear time.
\end{lemma}

Additionally to finding a maximum neighborhood ordering $\sigma=(v_1,\ldots,v_n)$, the algorithm in \cite{BCD98} also computes a maximum neighbor $m_i$ for each vertex $v_i$ such that for all $i < n$ no vertex $v_i$ is its own maximum neighbor ($v_i \neq m_i$). 
This is necessary for the next lemma.

\begin{lemma}\label{lem:ijInE_iff_mjInE2}
Let $G=(V,E)$ be a graph with $G^2=(V,E^2)$ and a maximum neighborhood ordering $\sigma=(v_1, \ldots, v_n)$ where for all $1\le i< n$, $m_i$ is a maximum neighbor of $v_i$ with $v_i \neq m_i$. 
If $1 \leq i < j \leq n$ and $m_i \neq v_j$, then $v_iv_j \in E^2 \Leftrightarrow m_iv_j \in E$.
\end{lemma}

\begin{proof}
$(\Leftarrow):$ Vertex $v_j$ is adjacent in $G$ to $m_i$ ($m_iv_j \in E$). 
Thus, the distance between $v_i$ and $v_j$ is at most $2$. 
So $v_i$ and $v_j$ are also adjacent in $G^2$ ($v_iv_j \in E^2$).

$(\Rightarrow):$ Vertices $v_i$ and $v_j$ are adjacent in $G^2$ ($v_iv_j \in E^2$).
If $v_iv_j \in E$, then $m_iv_j\in E$.
Now assume that $v_iv_j \notin E$.
Then vertices $v_i$ and $v_j$ have a common neighbor $v_k$ in $G$, choose the rightmost such vertex (that is,
the one maximizing the value of $k$).
We distinguish between two cases:
\begin{enumerate}[(i)]
    \item $i<k$.
    If $m_i= v_k$, then $m_iv_j \in E$ so we may assume that $m_i\neq v_k$.
    By definition $m_i$ is adjacent to all neighbors of $v_k$ in $G_{\sigma,i}$. 
    This includes $v_j$.

    \item $k<i$. In this case any maximum neighbor $m_k$ of $v_k$ satisfies $m_k\neq v_i$ (since $v_kv_j\in E$ but $v_iv_j\not \in E$), $v_im_k \in E$ and $m_kv_j \in E$. 
    In particular, $m_k$ is a common neighbor of $v_i$ and $v_j$. 
    Since $m_k = v_p$ for some $p>k$, this contradicts the choice of $v_k$.
\end{enumerate}
\end{proof}

\begin{algorithm}[ht]
\label{algo:EDdc}
\caption{A linear time algorithm for the \minwed problem in dually chordal graphs}

\KwIn{An instance $(G,w,k)$ of the \minwed problem where $G=(V,E)$ is dually chordal graph.}
\KwOut{An efficient dominating set $D$ in $G$ with $w(D)\le k$, if one exists, {\sc No}, otherwise.}

\BlankLine
{

    $\mu \leftarrow  \min\left\{\min\{w(v)\mid v\in V\},k\right\}$ \nllabel{line:posWeightsStart}

    \If{$\mu < 0$}
    {
        Solve the efficient domination problem on $G$.\nllabel{line:compgamma}

        \If{$G$ has no efficient dominating set}
        {
            \Return{\sc No};
        }
        \Else
        {
            Let $\gamma$ be the common cardinality of all efficient dominating sets of $G$.
        }

        For all $v \in V$ set $w(v) \leftarrow  w(v) + |\mu|$

        $k \leftarrow  k + |\mu| \cdot \gamma$
        \nllabel{line:posWeightsEnd}
    }

    $D \leftarrow  \emptyset$, $M \leftarrow  \max\{\sum_{x\in V}w(x),k\}+1$, $k'\leftarrow  M \cdot |V| - k$. \nllabel{line:compMk}

\ForAll{$v \in V$\nllabel{line2} }
    {
         Set $\omega(v) \leftarrow  M \cdot |N[v]| - w(v)$ and $\omega_p(v) \leftarrow  0$.\nllabel{line:compOmega}

        Set $v$ unmarked and not blocked.\nllabel{line4}
    }

    Find a maximum neighborhood ordering $\sigma=(v_1,\ldots,v_n)$
    of $G$ which simultaneously is a perfect elimination ordering of $G^2$,
    with the corresponding maximum
    neighbors $(m_1, \ldots, m_n)$ where $v_i \neq m_i$ for $1 \leq i < n$. \nllabel{line:compMNO}

    \For{$i =  1$ \KwTo $n$\nllabel{line6}}
    {
        For all $u \in N_{\sigma,i}[v_i]$ set $\omega(v_i) \leftarrow  \omega(v_i) - \omega_p(u)$\nllabel{line7}.

        \If{$\omega(v_i) > 0$\nllabel{line8} }
        {
             Mark $v_i$ and set $\omega_p(m_i) \leftarrow  \omega_p(m_i) + \omega(v_i)$.\nllabel{line9}
        }
    }

    \For{$i= n~{\bf downto}~1$\nllabel{line10} }
    {
        \If{$v_i$ is marked and $m_i$ is not blocked\nllabel{line11} }
        {
             Set $D \leftarrow D \cup \{v_i\}$ and block all $u \in N_{G}(v_i)$.\nllabel{line:compMWISend}
        }
    }

    \If{$\sum_{v \in D} \left(M \cdot |N[v]| - w(v)\right)\geq k'$\nllabel{line13} }
    {
        \Return{$D$};\nllabel{line14}
    }
    \Else{\nllabel{line15}
    \Return{\sc No};\nllabel{line16}
    }
}
\end{algorithm}

\begin{theorem}\label{thm:meed}
Algorithm~\ref{algo:EDdc} solves the \minwed problem for dually chordal graphs in linear time.
\end{theorem}

\begin{proof}

To get Algorithm~\ref{algo:EDdc} we extended the Algorithm~\ref{algo:mwisChordal} by lines~\ref{line:posWeightsStart}--\ref{line:posWeightsEnd} to ensure non-negative weights. 
Additionally, we added the calculation for~$M$ and~$k'$ (line~\ref{line:compMk}), and defined $\omega(v)$ now as $M \cdot |N[v]| - w(v)$ instead of $|N[v]|$ (line~\ref{line:compOmega}). 
Therefore, based on Lemma~\ref{lem:wed-mwis}, we first created an instance $(G,w,k)$ for the \minwed problem with non-negative weights, and then by Lemma~\ref{lem:wed-mwis-2} an instance $(G,w',k')$ for the \mwis problem on $G^2$ (with $w'(v)=\omega(v)$).

Next, the \mwis instance is solved in the lines~\ref{line:compMNO}--\ref{line:compMWISend} in linear time.
To achieve this, we rewrite Algorithm~\ref{algo:mwisChordal} to work on the dually chordal graph instead of its chordal square.

Based on Lemma~\ref{lem:MNO_linear}, line~\ref{line:compMNO} of Algorithm~\ref{algo:EDdc} computes a perfect elimination ordering of $G^2$.

Let $v_i$ and $v_j$ be adjacent in $G^2$ and $i <j$. Now Lemma~\ref{lem:ijInE_iff_mjInE2} allows to modify the two loops. 
For the first loop (line~\ref{line6}), there is an extra vertex weight $\omega_p$. 
When $v_i$ is marked, instead of decrementing the weights $\omega$ of the neighbors of $v_i$ (and their neighbors), $\omega_p$ of $v_i$'s maximum neighbour~$m_i$ is incremented by $\omega(v_i)$.
Now before comparing $\omega(v_j)$ to $0$, $\omega(v_j)$ is decremented by $\omega_p(u)$ for all $u \in N[v_j]$. 
Because $m_i$ is adjacent to $v_j$ (Lemma~\ref{lem:ijInE_iff_mjInE2}), this ensures that each time the weight~$\omega(v_j)$ is compared to $0$, it has the same value as it would have in Algorithm~\ref{algo:mwisChordal}.

For the second loop (line~\ref{line10}) the argumentation works similarly. 
After selecting a vertex~$v_j$ (i.e., $v_j \in D$), all its neighbors in $G^2$ are blocked. 
Thus by Lemma~\ref{lem:ijInE_iff_mjInE2}, if a vertex~$v_i$ is adjacent in $G^2$ to a selected vertex, then the maximum neighbor $m_i$ is blocked.

It follows that after line~\ref{line:compMWISend}, $D$ is a solution to the instance $(G,w',k')$ for the \mwis problem on $G^2$ created earlier.
    
Thus, the set~$D$ is a solution for the given \minwed instance if and only if $\sum_{v \in D} \big(M \cdot |N[v]| - w(v)\big) \geq k'$ (lines~\ref{line13}--\ref{line16}).

\medskip
Solving the efficient domination problem (line~\ref{line:compgamma}) and finding a maximum neighborhood ordering (line~\ref{line:compMNO}) can be done in linear time \cite{BCD98,BLR}. 
Also the overall runtime of lines~\ref{line6}--\ref{line:compMWISend} is bounded by the number of vertices and edges in $G$. 
Thus, Algorithm~\ref{algo:EDdc} runs in linear time.
\end{proof}

As explained at the end of the proof of Theorem~\ref{thm:wed-mwis}, the \maxwed problem can be reduced in time ${\mathcal O}(n+m)$ to the \minwed problem by negating the weights. 
Therefore, the \maxwed problem can also be solved in linear time on dually chordal graphs.

\subsection{AT-free graphs}
\label{sec:AT-free}

We now consider the class of AT-free graphs. 
Recall the following result due to Chang {\it et al.}~\cite{CHK03}.

\begin{theorem}[\cite{CHK03}]\label{ATGk}
Every proper power of an AT-free graph is a cocomparability graph.
\end{theorem}

Therefore, by Theorem~\ref{thm:wed-mwis}, a polynomial time algorithm for the \minwed and \maxwed problems in AT-free graphs will follow if the \mwis problem is solvable in polynomial time in the class of cocomparability graphs.
The \mwis problem in the class of cocomparability graphs is equivalent to the maximum weight clique problem in the class of comparability graphs, for which the following result is known.

\begin{theorem}[\cite{Gol04,MS99}]
The maximum weight clique problem can be solved in linear time on comparability graphs.
\end{theorem}

The algorithm given by Golumbic~\cite{Gol04} requires a transitive orientation of a comparability graph $G$ as input. 
Such an orientation can be found in linear time using methods of McConnell and Spinrad~\cite{MS99}. Consequently:

\begin{theorem}\label{MWIS:cocomp}
The \mwis problem on instances $(G,w,k)$ such that $G$ is a cocomparability graph can be solved in time ${\mathcal O}(|V(G)|+|E(\overline{G})|)$.
\end{theorem}

Theorems~\ref{thm:wed-mwis}, \ref{ATGk} and~\ref{MWIS:cocomp} imply the following result.

\begin{theorem}\label{thm:AT-free}
The \minwed and \maxwed problems are solvable in the class of AT-free graphs in time ${\mathcal O}\big(\min\{nm+n^2,n^\omega\}\big)$.
\end{theorem}

\begin{proof}
Let ${\cal C}$ be the class of AT-free graphs, and let $G\in {\cal C}$ with $n = |V(G)|$. By Theorem~\ref{ATGk}, the square of $G$ is cocomparability.
By Theorem~\ref{MWIS:cocomp}, the \mwis problem is solvable on $G^2$ in time ${\mathcal O}(n^2)$. Thus, Theorem~\ref{thm:wed-mwis} implies that in the class of AT-free graphs, the \maxwed and \minwed problems are solvable in time
${\mathcal O}\big(\min\{nm+n,n^\omega\}+ n^2\big) = {\mathcal O}\big(\min\{nm+n^2,n^\omega\}\big)$.
\end{proof}

Theorem~\ref{thm:AT-free} generalizes the polynomial-time algorithms for the efficient dominating set problem in the class of AT-free graphs by Brandst\"adt \emph{et al.}~\cite{BLR} and by Broersma \emph{et al.}~\cite{BKM99}.
Both papers~\cite{BLR,BKM99} solve the unweighted version of the problem in time ${\mathcal O}(n^4)$, while we give an algorithm of complexity ${\mathcal O}\big(\min\{nm+n^2,n^\omega\}\big)$ to solve the more general, weighted versions of the problem.
The polynomial time solvability implied by Theorem~\ref{thm:AT-free} can also be seen as a common extension of the polynomial-time solvability of the \minwed problem on interval graphs~\cite{CL94}, cocomparability graphs~\cite{CRC95}, and permutation graphs~\cite{LLT97}, all subclasses of AT-free graphs.

\subsection*{Acknowledgements}

We are grateful to the two anonymous referees for their helpful comments. 
M.M.~is supported in part by ``Agencija za raziskovalno dejavnost Republike Slovenije'', research program P$1$--$0285$ and research projects J$1$-$5433$, J$1$-$6720$, and J$1$-$6743$.

\bibliographystyle{abbrv}

\end{document}